\documentclass[journal]{IEEEtran}
\usepackage{packages}
\def \a{\bm{a}}
\graphicspath{{Figures/}}
\usepackage{amssymb}
\usepackage{balance}
\usepackage{mathtools}
\DeclarePairedDelimiter\floor{\lfloor}{\rfloor}

\begin{document}

\title{Tight Lower Bound on the Tensor Rank based on the Maximally Square Unfolding}

\author{Giuseppe G. Calvi, \textit{Student member, IEEE}, Bruno Scalzo Dees,  \textit{Student member, IEEE},  \\Danilo P. Mandic, \textit{Fellow, IEEE}}


\maketitle

\begin{abstract}
Tensors decompositions are a class of tools for analysing datasets of high dimensionality and variety in a natural manner, with the Canonical Polyadic Decomposition (CPD) being a main pillar. While the notion of CPD is closely intertwined with that of the tensor rank, $R$, unlike the matrix rank, the computation of the tensor rank is an NP-hard problem, owing to the associated computational burden of evaluating the CPD. To address this issue, we investigate tight lower bounds on $R$ with the aim to provide a reduced search space, and hence to lessen the computational costs of the CPD evaluation. This is achieved by establishing a link between the maximum attainable lower bound on $R$ and the dimensions of the matrix unfolding of the tensor with aspect ratio closest to unity (maximally square). Moreover, we demonstrate that, for a generic tensor, such lower bound can be attained under very mild conditions, whereby the tensor rank becomes detectable. Numerical examples demonstrate the benefits of this result.

\end{abstract}


\IEEEpeerreviewmaketitle

\section{Introduction}

The increasing prominence of multisensor technology and the associated generated data quantities of exceedingly high dimensionality have highlighted the limitations of standard ``flat-view" matrix and vector models, both in terms of their accuracy and computational requirements \cite{Kolda2009, Mandic2015}. This has, in turn, motivated the development of new, sophisticated tools capable of coping with the sheer complexity of modern datasets. Tensors are multidimensional generalizations of matrices and vectors, which benefit from the power of the underpinnning multilinear algebra to flexibly and efficiently account for multi-way relationships among data. Owing to their ability to exploit the underlying latent data structures, tensors have both pushed theoretical performance limits \cite{Kolda2001, DeLathauwer2000_2, DeLathauwer2000} and opened new avenues for practical applications \cite{Mandic2016_2, Mandic2017, Comon2014, Sidiropoulos2017}, ranging from scientific computing and physics through to signal processing and machine learning.

However, the typically high dimensionality of tensors implies that the application of common algorithms on data in the raw tensor format may become intractable due to the \textit{curse of dimensionality} \cite{Mandic2015}. To tackle this issue, in 1927 Hitchcock  was the first to propose the concept tensor decompositions, and in particular, that of the polyadic expansion of a tensor \cite{Hitchcock1927} (i.e. a sum of rank-one terms), which yields a low-rank approximation of the original tensor through factors of lower complexity. This methodology became popular in the 1970s after its adoption by the psychometrics community, under the name of CANonical DECOMPosition (CANDECOMP) \cite{Carroll1970} or PARAllel FACtors (PARAFAC) \cite{Harshman1970}. Its most common name at present is  Canonical Polyadic Decomposition (CPD), which has maintained its role as one of the workhorses of tensor decompositions  \cite{Tucker1963, Tucker1964, Oseledets2011}. Indeed, the CPD has been widely used as an advanced tool for signal separation within the signal processing and data analytics communities, such as in audio and speech processing, biomedical engineering, machine learning, and chemometrics \cite{Acar2009, Smilde2004}. Further applications include  wireless communications where, owing to a rank-one structure of the real or complex exponentials involved, making use of the CPD natural \cite{Nikias1993,  Sorensen2013}.

Despite its virtues, the CPD faces a fundamental problem related to critically relying on a given value attributed to the tensor rank, $R$. Therefore, the first issue which arises when computing a CPD is the choice of $R$,  however, determining the rank of a tensor is an NP-hard problem \cite{Haastad1990}. Although the computation can be performed successively for increasing values of $R$, until the model fits the underlying tensor arbitrarily well, it is a well-known fact that the best low-rank tensor approximation may not always exist, implying that a value for $R$ has to be arbitrarily chosen \cite{Kolda2009}. Because the CPD is intrinsically multilinear, the most common approach for the computation of the CPD is through an Alternating Least Squares (ALS) procedure, whereby the CP factors are estimated one at a time, while the others are kept fixed. The presence of local minima in this approach, combined with the fact that an acceptable value for $R$ has to be sought iteratively, makes CPD a rather computationally expensive procedure. 

This motivates us to propose ways to relax the computational burden associated with the CPD, by studying the properties of tight lower bounds on the tensor rank through matrix unfoldings of the associated tensor. Knowledge of lower bounds on $R$ would, in turn, lower the computational cost by reducing the search space for $R$. In this work, we show that the maximum of such bounds is equal to the maximum matrix rank attainable by any unfolding of the tensor. To show this, we start from the definition of a ``flattened CPD", to allow for a manipulation of the CP factors using standard linear algebra; this serves as a basis to develop the theory behind the proposed tight lower bound of the CPD rank which relies on a specific permutation of the original tensor. From Sylvester's rank inequality, we subsequently demonstrate that, for generic tensors, the bound is attained under very mild conditions. This, in turn, allows for the tensor rank to be detected by simply inspecting the matrix rank of the associated tensor unfolding. We illustrate the implications of this result through numerical examples, which highlight its practical benefits.

\pagebreak

\section{Notation and Background}

We adopt most of our notation from \cite{Kolda2009}.

\begin{table}[H]
	\centering
	\caption{Main tensor nomenclature.}
	\label{table:nomenclature}
	
	\begin{tabular}{ll}
		\hline
		
		\vspace{-2mm}	& \\
		
		$\ten{X} \in \mathbb{R}^{I_1 \times I_2 \times \cdots \times I_N}$ & \begin{tabular}[c]{@{}l@{}} Tensor of order $N$ of\\ size $I_1 \times I_2 \times \cdots \times I_N$\end{tabular} \vspace{2mm} \\

		$x_{i_1i_2\cdots i_N}=\ten{X}(i_1,i_2,\dots,i_N)$ 	& $(i_1,i_2,\dots,i_N)$ entry of $\ten{X}$ \vspace{2mm} \\

		$x$, $\mathbf{x}$, $\mathbf{X} $	& Scalar, vector, matrix  \vspace{2mm}\\

		$\mathbf{A}^{(n)}$ & \begin{tabular}[c]{@{}l@{}}Factor matrices \\  \end{tabular} \vspace{2mm} \\
		
		$(\cdot)^T$ & \begin{tabular}[c]{@{}l@{}} Transpose  operator \\  for matrices\end{tabular} \vspace{2mm} \\

		
		$\circ,  \odot, \otimes$	&  \begin{tabular}[c]{@{}l@{}} Outer, Khatri-Rao, and \\ Kronecker products\end{tabular} \vspace{2mm}  \\ 
		
		$||\cdot||_F$	& Frobenius norm  \vspace{2mm}  \\

		\hline
	\end{tabular}
\end{table}

A Khatri-Rao product of matrices $\mathbf{A} = [\mathbf{a}_1, \dots, \mathbf{a}_R] \in \mathbb{R}^{I \times R}$ and $\mathbf{B} = [\mathbf{b}_1, \dots, \mathbf{b}_R] \in \mathbb{R}^{J \times R}$ yields $\mathbf{C} \in \mathbb{R}^{IJ \times R}$ with columns $\mathbf{c}_i = \mathbf{a}_i\otimes \mathbf{b}_i$, $i=1, \dots, R$, where the symbol $\otimes$ denotes the Kronecker product. Furthermore, the outer product of an $N$-th order tensor $\ten{A} \in \mathbb{R}^{I_1 \times I_2 \times \cdots \times I_N}$ and an $M$-th order tensor $\ten{B} \in \mathbb{R}^{J_1 \times J_2 \times \cdots \times J_M}$ yields an $(N+M)$-th order tensor $\ten{C} \in \mathbb{R}^{I_1 \times \cdots \times I_N \times J_1 \times \cdots \times J_M}$ with entries $c_{i_1\cdots i_N j_1 \cdots j_M} = a_{i_1 \cdots i_N} b_{j_1 \cdots j_M}$.

\begin{definition}
	An $N$-th order tensor is said to be of \textbf{rank-1} if it can be written as an outer product of $N$ vectors, that is
	\begin{equation}
	\ten{X} = \a^{(1)} \circ \a^{(2)} \circ \dots \circ \a^{(N)}
	\end{equation}
	In other words, each element of such tensor is a product of the corresponding vector elements
	\begin{equation}
	x_{i_1i_2\cdots i_N} = a_{i_1}^{(1)} a_{i_2}^{(2)}\cdots  a_{i_N}^{(N)}, \text{  for all  }  1\leq i_n \leq I_n
	\end{equation}
\end{definition}

The Canonical Polyadic Decomposition (CPD) decomposes any tensor into a sum of rank-1 tensors. For instance, the CPD of a $3$-rd order tensor, $\ten{X} \in \mathbb{R}^{I_1 \times I_2 \times I_3}$, yields
\begin{equation}\label{eq:cpd3}
\ten{X} \approx \sum_{r=1}^{R} \lambda_r \mathbf{a}^{(1)}_r \circ \mathbf{a}^{(2)}_r \circ \mathbf{a}^{(3)}_r
\end{equation}
where $\mathbf{a}^{(1)} \in \mathbb{R}^{I_1}$, $\mathbf{a}^{(2)} \in \mathbb{R}^{I_2}$, $\mathbf{a}^{(3)} \in \mathbb{R}^{I_3}$. If the approximation in (\ref{eq:cpd3}) is exact, then $R$ is said to be the rank of the tensor.

\begin{definition}
	The \textbf{rank} of a tensor $\ten{X}$, $\rank(\ten{X})=R$, is the smallest number of rank-1 tensor factors in (\ref{eq:cpd3})  which ensures that their sum  generates $\ten{X}$ exactly.
\end{definition}

The elementwise form of (\ref{eq:cpd3}) can be written as
\begin{equation}
\ten{X}(i_1, i_2, i_3) = \sum_{r=1}^{R} \lambda_r a^{(1)}_{i_1 r} a^{(2)}_{i_2 r} a^{(3)}_{i_3 r}
\end{equation}
If the component vectors are combined into matrices, e.g. $\mathbf{A}^{(1)} = [\mathbf{a}^{(1)}_1, \mathbf{a}^{(1)}_2, \dots, \mathbf{a}^{(1)}_R]$, the CPD can be expressed through multi-linear products (see \cite{Kolda2009} for further detail) as
\begin{equation}
\begin{aligned}
\ten{X} & \approx \ten{D} \times_1 \mathbf{A}^{(1)} \times_2 \mathbf{A}^{(2)} \times_3 \mathbf{A}^{(3)}\\
& \equiv \llbracket \ten{D}; \mathbf{A}^{(1)}, \mathbf{A}^{(2)}, \mathbf{A}^{(3)} \rrbracket
\end{aligned}
\end{equation}
where $\mathbf{A}^{(1)} \in \mathbb{R}^{I_1 \times R}$, $\mathbf{A}^{(2)} \in \mathbb{R}^{I_2 \times R}$, $\mathbf{A}^{(3)} \in \mathbb{R}^{I_3 \times R}$, and $\ten{D} \in \mathbb{R}^{R\times R\times R}$ is a diagonal tensor containing $\lambda_r$'s on its diagonal, that is 
\begin{equation}\label{eq:core}
\ten{D}(r_1, r_2, r_3) = 
\begin{cases}
\lambda_r, \text{ if } r_1=r_2=r_3=r \\
0, \text{ otherwise}
\end{cases}
\end{equation}
The CPD for a $3$-rd order tensor is illustrated in Figure \ref{fig:cpd3}.

\begin{figure}[t]
	\centering
	\includegraphics[width=0.9\linewidth,  trim={0cm 0cm 0cm 0cm}, clip]{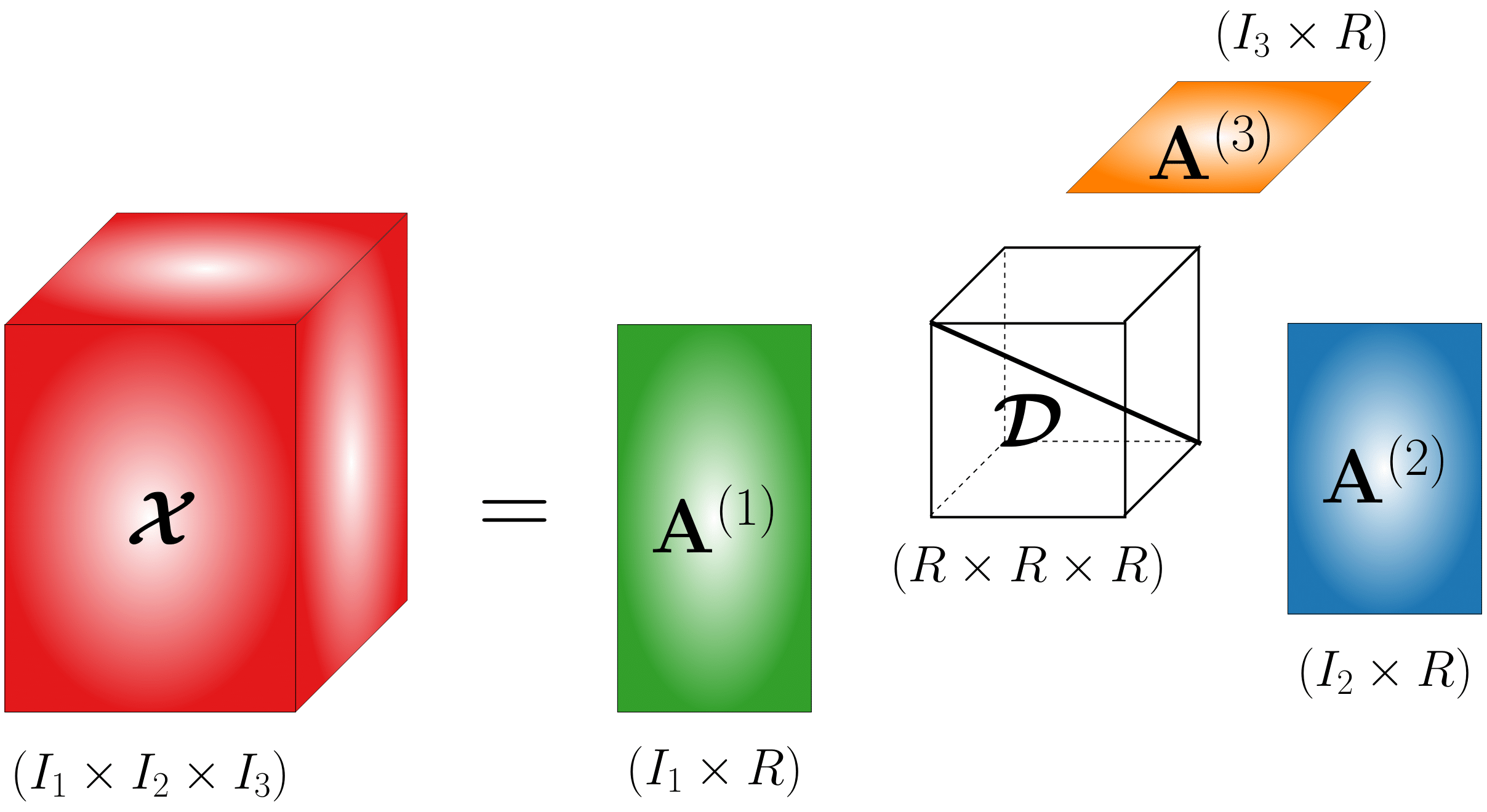}
	\caption{{Illustration of the CPD for a $3$-rd order tensor $\ten{X}\in\mathbb{R}^{I_1\times I_2 \times I_3}$.}}
	\label{fig:cpd3}
\end{figure}

The CPD for an $N$-th order tensor, $\ten{X} \in \mathbb{R}^{I_1 \times I_2 \times \cdots \times I_N}$,  takes a general form 
\begin{equation}\label{eq:cpd}
\begin{aligned}
\ten{X} & \approx \sum_{r=1}^{R} \lambda_r \mathbf{a}_r^{(1)} \circ \mathbf{a}_r^{(2)} \circ \cdots \circ \mathbf{a}_r^{(N)}\\
&= \ten{D} \times_1 \mathbf{A}^{(1)} \times_2 \mathbf{A}^{(2)} \times \cdots \times_N \mathbf{A}^{(N)}\\
&\equiv \llbracket  \ten{D}; \mathbf{A}^{(1)}, \mathbf{A}^{(2)}, \dots, \mathbf{A}^{(N)} \rrbracket
\end{aligned}
\end{equation}

\begin{remark}
	If the approximation in (\ref{eq:cpd}) is exact, in this work, we say that the underlying tensor, $\ten{X}\in \mathbb{R}^{I_1 \times I_2 \times \cdots \times I_N}$ \textit{admits a CPD}.
\end{remark}

\section{Tight Lower Bound Derivation}

To arrive at a tight lower bound on tensor rank, $R$, we start from the following definitions.

\begin{definition}\label{def:n-unfolding}
	The \textbf{$n$-unfolding} of a tensor, $\ten{X} \in \mathbb{R}^{I_1 \times I_2 \times \cdots \times I_N}$, is denoted by $\mathcal{X}_{< n >}$, and represents a matrix with entries
	\begin{equation}
	(\mathbf{X}_{ < n > })_{\overline{i_1 \dots i_n},\overline{i_{n+1} \dots i_N   }} = x_{i_1 i_2 \dots i_N}
	\end{equation}
\end{definition}

\begin{definition}
	The \textbf{$n$-rank} of a tensor, $\ten{X} \in \mathbb{R}^{I_1 \times I_2 \times \cdots \times I_N}$, is $\textnormal{\rank}(\mathbf{X}_{<n>})$, i.e. the matrix rank of its $n$-unfolding. 
\end{definition}


Note that if an $N$-th order tensor admits a CPD, then its $n$-mode unfolding is 
\begin{equation}\label{eq:cpd_unfold_2}
	\begin{aligned}
		\mathbf{X}_{<n>} &= (\mathbf{A}^{(n)} \odot \cdots \odot \mathbf{A}^{(1)} ) \mathbf{D} (\mathbf{A}^{(N)}  \odot \cdots \odot \mathbf{A}^{(n+1)}  )^T \\
		& =  \Big(\Modot_{i=n}^{^1} \mathbf{A}^{(i)}\Big) \mathbf{D} \Big(\Modot_{i=N}^{^{n+1}} \mathbf{A}^{(i)T}\Big)
	\end{aligned}
\end{equation}
Using the fact that, for matrices, $\rank(\mathbf{ABC}) \leq \min \{ \rank(\mathbf{A}), \rank(\mathbf{B}), \rank(\mathbf{C})\}$, and that $\mathbf{D} \in \mathbb{R}^{R\times R}$ is diagonal, it then follows that
\begin{equation}
	\begin{aligned}
		&\rank (\mathbf{X}_{<n>})  \leq\\
		& \min \Big\{\rank \Big(  \Modot_{i=n}^{^1} \mathbf{A}^{(i)T} \Big), \rank \Big(  \Modot_{i=N}^{^{n+1}} \mathbf{A}^{(i)T} \Big),    \rank(\mathbf{D})   \Big\} \\
		& =  \min \Big\{\rank \Big(  \Modot_{i=n}^{^1} \mathbf{A}^{(i)T} \Big), \rank \Big(  \Modot_{i=N}^{^{n+1}} \mathbf{A}^{(i)T} \Big),  R  \Big\} \\
		& \leq R
	\end{aligned}
\end{equation}
Because $\rank(\mathbf{X}_{<n>}) \leq R$ for all $n=1, \dots, N$, this implies that the rank of the tensor, $R$, is bounded from below by its maximum $n$-rank, that is
\begin{equation}\label{eq:lowerbound1}
	\max_{n=1, \dots, N}\{\rank(\mathbf{X}_{ < n >})\} \leq R
\end{equation}

\noindent Let $k$ denote the argument index of the $n$-unfolding of tensor $\ten{X} \in \mathbb{R}^{I_1 \times I_2 \times \cdots \times I_N}$ with maximum $n$-rank, that is
\begin{equation}
	k= \argmaxA_{n=1, \dots, N}\{  \rank(\mathbf{X}_{<n>}) \}
\end{equation} 
and consider now the $k$-unfolding of $\ten{X}$, 
\begin{equation}\label{eq:kron}
	\mathbf{X}_{<k>} = \underbrace{(\mathbf{A}^{(k)} \odot \cdots \odot \mathbf{A}^{(1)})}_{\mathbf{\underline{A}}_1 \in \mathbb{R}^{I_1\cdots I_k \times R} }  \mathbf{D} {\underbrace{(\mathbf{A}^{(N)} \odot \cdots \odot \mathbf{A}^{(k+1)})}_{{\mathbf{\underline{A}}_2 \in \mathbb{R}^{I_{k+1}\cdots I_N \times R} } } }  ^T
\end{equation}
From Sylvester's rank inequality \cite{Matsaglia1974}, for any matrix triplet $\mathbf{A} \in \mathbb{R}^{N\times M}$, $\mathbf{B} \in \mathbb{R}^{M \times P}$, $\mathbf{C} \in \mathbb{R}^{P\times K}$, we have the following 
\begin{equation}
	\rank( \mathbf{ABC} ) \geq \rank(\mathbf{A}) +  \rank(\mathbf{B}) + \rank(\mathbf{C}) - M - P
\end{equation}
It then follows that, from (\ref{eq:kron}), we obtain
\begin{equation}\label{eq:bound1}
	\rank( \mathbf{X}_{<k>}) \geq R_1 + R + R_2 - R - R
\end{equation}
where $R_1 = \rank( \mathbf{\underline{A}}_1)$, $R_2 = \rank( \mathbf{\underline{A}}_2)$. Combining the above result with that in (\ref{eq:lowerbound1}), yields
\begin{equation}\label{eq:squeeze}
	 R_1 + R_2 - R \leq \rank(\mathbf{X}_{<k>}) \leq R
\end{equation}
whereby if the equality $ R_1 + R_2 - R = R$ is verified, then $\rank(\mathbf{X}_{<k>}) = R$. Consequently, for this to hold, we require $R_1 = R_2 = R$. Generically\footnote{We adopt the same terminology as in \cite{DeLathauwer2008_1, DeLathauwer2008_2}. When we say that, generically, $\rank (\mathbf{A} \odot \mathbf{B}) = \min\{ IJ, R  \} $, for $\mathbf{A} \in \mathbb{R}^{I \times R}, \mathbf{B} \in \mathbb{R}^{J\times R}$ we mean that the property should hold with probability one for matrices $\mathbf{A}$, $\mathbf{B}$, drawn from a continuous probability density function over their respective subspaces.}, for matrices $\mathbf{A} \in \mathbb{R}^{I \times R}$, and $\mathbf{B} \in \mathbb{R}^{J \times  R}$, we have the identity $\rank (\mathbf{A} \odot \mathbf{B}) = \min\{ IJ, R  \} $ \cite{DeLathauwer2008_1, DeLathauwer2008_2}. As a result, for generic $ \underline{\mathbf{A}}_1$ and $ \underline{\mathbf{A}}_2$, we have that $\rank( \underline{\mathbf{A}}_1 )  = \min\{\prod_{i=1}^{k} I_i, R\}$ and $\rank( \underline{\mathbf{A}}_2 )  = \min\{\prod_{i=k+1}^{N} I_i, R\}$. This implies, in turn, that $R_1 = R_2 = R$ is equivalent to satisfying $R \leq \prod_{i=1}^{k} I_i$ and $R \leq \prod_{i=k+1}^{N} I_i$ simultaneously, that is
\begin{equation}\label{eq:lowerbound2}
	R \leq \min \bigg\{   \prod_{i=1}^{k} I_i,  \prod_{i=k+1}^{N} I_i   \bigg\}
\end{equation}
In other words, if a tensor, $\ten{X} \in \mathbb{R}^{I_1 \times I_2 \times \cdots \times I_N}$, admits a CPD and if (\ref{eq:lowerbound2}) holds, then $R = \rank(\mathbf{X}_{<k>})$, where $k= \argmaxA_{n}\{  \rank(\mathbf{X}_{<n>}) \}$. Moreover, notice that the strict inequality
\begin{equation}\label{eq:bound}
	\begin{aligned}
		&\rank(\mathbf{X}_{<k>}) < \bigg\{   \prod_{i=1}^{k} I_i,  \prod_{i=k+1}^{N} I_i   \bigg\}
	\end{aligned}
\end{equation}
implies 
\begin{equation}\label{eq:bound2}
\begin{aligned}
R < \min \bigg\{   \prod_{i=1}^{k} I_i,  \prod_{i=k+1}^{N} I_i   \bigg\}
\end{aligned}
\end{equation}
and, as a result, (\ref{eq:lowerbound2}) is satisfied. Thus, if $\mathbf{X}_{<k>}$ is rank-deficient (\ref{eq:bound2}), generically we have that $\rank(\mathbf{X}_{<k>}) = R$. 

In the following, we further show that it possible to find a value $r$ such that $\rank(\mathbf{X}_{<k>}) \leq r = R$. This is done by considering permutations of the indices of the original tensor $\ten{X}$.

\begin{definition}
	For a tensor, $\ten{X} \in \mathbb{R}^{I_1 \times I_2 \times \cdots \times I_N}$, its $N$ indices admit $N!$ permutations, denoted by $\rho=1, \dots, N!$, with the corresponding permuted tensor, $\ten{X}^\rho \in \mathbb{R}^{[I_1 \times I_2 \times \cdots \times I_N]^\rho}$. The permuted tensor has the same elements as $\ten{X}$, but with their indices permuted according to a permutation scheme $\rho$. Therefore, $\ten{X}^\rho$ is referred to as a \textbf{permutation} of $\ten{X}$.
\end{definition}

For illustration, consider a $3$-rd order tensor, $\ten{X} \in \mathbb{R}^{I_1 \times I_2 \times I_3}$. Based on its indices $\{I_1, I_2, I_3 \}$, we can define up to $\rho=1, \dots, 3!=6$ permutations, which include $\ten{X}^1 \in \mathbb{R}^{[I_1 \times I_2 \times I_3]^1} \equiv \ten{X}^1 \in \mathbb{R}^{I_1 \times I_2 \times I_3}$, or $\ten{X}^2 \in \mathbb{R}^{[I_1 \times I_2 \times I_3]^2} \equiv \ten{X}^2 \in \mathbb{R}^{I_2 \times I_1 \times I_3}$. 
 
\begin{definition}\label{def:maximal}
	The argument index of the $n$-unfolding  for a permutation $\rho$ which exhibits a maximum $n$-rank is denoted by $k_\rho$, that is
	\begin{equation}
	k_\rho = \argmaxA_{n=1, \dots, N} \{\textnormal{\rank}(\mathbf{X}^\rho_{<n>})  \}, \text{ } \rho = 1, \dots, N!
	\end{equation}
	By continuity, a \textbf{maximum $n$-rank permutation} (or \textbf{maximal permutation}, for brevity) of tensor $\ten{X} \in \mathbb{R}^{I_1 \times I_2\times \cdots \times I_N}$ is denoted by $\ten{X}^{\rho*}$, where $\rho^*$ satisfies
	\begin{equation}
	\rho^* = \argmaxA_{\rho=1, \dots, N!} k_\rho
	\end{equation}
\end{definition}

\begin{remark}
	A tensor, $\ten{X} \in \mathbb{R}^{I_1 \times I_2 \times \cdots \times I_N}$, may admit multiple maximal permutations.
\end{remark}

\begin{remark}\label{prop:cpdperm}
	If a tensor, $\ten{X} \in \mathbb{R}^{I_1 \times I_2 \times \cdots \times I_N},$ admits a CPD of rank $R$, then any permutation of $\ten{X}$, given by $\ten{X}^\rho \in \mathbb{R}^{[I_1 \times I_2 \times \cdots \times I_N]^\rho}$, admits a CPD of rank $R$ which shares the same factors with the CPD of $\ten{X}$, but rearranged following the permutation scheme of the indices. In other words, if
	\begin{equation}
	\ten{X} = \llbracket \mathbf{D}; \mathbf{A}^{(1)}, \dots, \mathbf{A}^{(N)} \rrbracket
	\end{equation}
	then, 
	\begin{equation}
	\ten{X}^\rho = \llbracket \mathbf{D}; [\mathbf{A}^{(1)}, \dots, \mathbf{A}^{(N)}]^\rho \rrbracket
	\end{equation}
	and $\rank(\ten{X}) = \rank(\ten{X}^\rho) = R$. Note that, defining $\mathcal{I} = \{I_1, I_2, \dots, I_N\}$ and $\mathcal{I}^{\rho*} = \mathcal{J} = \{J_1, J_2, \dots, J_N\}$, the maximal permutation, $\rho^*$, can be interpreted as the bijection $\rho^*: \mathcal{I} \mapsto \mathcal{J}$.
\end{remark}

Equipped with the above, this section can be summarized in the following,
\begin{theorem}\label{th:rank}
		If a tensor, $\ten{X} \in \mathbb{R}^{I_1 \times I_2 \times \cdots \times I_N}$, admits a CPD of rank $R$, its maximal permutation, $\ten{X}^{\rho*} \in \mathbb{R}^{J_1 \times J_2 \times \cdots \times J_N}$, also admits a CPD of rank $R$. Hence, $R$ is bounded from below by the rank of the $k$-unfolding, $\mathbf{X}_{<k>}^{\rho*}$, with $k = \argmaxA_{n} \rank(\mathbf{X}_{<n>}^{\rho*})$, that is
		\begin{equation}
			\rank(\mathbf{X}_{<k>}^{\rho*}) \leq R
		\end{equation}
		Moreover, if $\mathbf{X}_{<k>}$ is rank-deficient, i.e. if 
		\begin{equation}\label{eq:th}
			\rank(\mathbf{X}_{<k>}^{\rho*}) < \min \bigg\{   \prod_{i=1}^{k} J_i,  \prod_{i=k+1}^{N} J_i   \bigg\}
		\end{equation}
		then, generically, $\rank(\mathbf{X}_{<k>}^{\rho*}) = \rank(\ten{X}^{\rho*}) = R$.
\end{theorem} 
\begin{proof}
	By (\ref{eq:lowerbound1}) and Remark \ref{prop:cpdperm}, we obtain $\rank(\mathbf{X}_{<k>}^{\rho*}) \leq R$, while the rest of Theorem \ref{th:rank} follows from (\ref{eq:squeeze})--(\ref{eq:bound}).
\end{proof}

\section{Rank Detectability}
Theorem \ref{th:rank} implies that it is generically possible to \textit{detect} the rank of a tensor $\ten{X} \in \mathbb{R}^{I_1 \times I_2 \times \cdots \times I_N}$ simply by computing the maximum $n$-rank of its maximal permutation, provided that (\ref{eq:th}) holds, in which case, the rank of the tensor is automatically given as the rank of the unfolding. The problem then boils down to finding the maximal permutation $\rho^*$, which is generically given as the solution to,
\begin{equation}\label{eq:problem}
    \argmaxA_{\mathcal{S}_1, \mathcal{S}_2} \min \bigg\{    \prod_{i \in \mathcal{S}_1} I_i, \prod_{i \in \mathcal{S}_2} I_i  \bigg\}
\end{equation}
where $\mathcal{S}_1 \cup \mathcal{S}_2 = \{I_1, \dots, I_N \}$, $\mathcal{S}_1 \cap \mathcal{S}_2 = \varnothing$, $|\mathcal{S}_1| = N_1$, $|\mathcal{S}_2| = N_2$, and $N_1+N_2 = N$. Intuitively, the problem in (\ref{eq:problem}) is equivalent to that of finding the matrix unfolding of the original tensor $\ten{X} \in \mathbb{R}^{I_1 \times I_2 \times \cdots \times I_N}$ which is ``as square as possible", referred to as \textbf{maximally square unfolding}.

The problem in (\ref{eq:problem}) can hence be expressed as determining the unfolding with aspect ratio closest to unity, that is
\begin{equation}
	\argminA_{\mathcal{S}_1, \mathcal{S}_2} \bigg| \frac{  \prod_{i \in \mathcal{S}_1}I_i  }{\prod_{i \in \mathcal{S}_2}I_i} -1 \bigg|
\end{equation}
which is equivalent to
\begin{equation}\label{eq:step}
	\begin{aligned}
		&\argminA_{\mathcal{S}_1, \mathcal{S}_2}  \bigg| \log \bigg( \frac{  \prod_{i \in \mathcal{S}_1}I_i  }{\prod_{i \in \mathcal{S}_2}I_i} \bigg) \bigg|\\
		=&\argminA_{\mathcal{S}_1, \mathcal{S}_2} \bigg|   \log \Big(\prod_{i \in \mathcal{S}_1}I_i\Big) -  \log \Big(\prod_{i \in \mathcal{S}_1}I_i\Big) \bigg|\\
		=&\argminA_{\mathcal{S}_1, \mathcal{S}_2} \bigg| \sum_{i \in \mathcal{S}_1} \log(I_i) - \sum_{i \in \mathcal{S}_2} \log(I_i)      \bigg|
	\end{aligned}
\end{equation}
Because all $I_i$'s are positive integers and the log function is monotonic, the solution to (\ref{eq:step}) is also the solution to 
\begin{equation}
	\argminA_{\mathcal{S}_1, \mathcal{S}_2} \bigg| \sum_{i \in \mathcal{S}_1} I_i - \sum_{i \in \mathcal{S}_2} I_i      \bigg|
\end{equation}
The above is known as the \textit{number partitioning problem}, and has been extensively studied in the literature \cite{Mertens1998, Ferreira1998, Mertens2006}. Although the partitioning problem is NP-complete, it can be solved efficiently in practice via dynamic programming. The optimal sets $\mathcal{S}_1$, $\mathcal{S}_2$ immediately give the maximal unfolding $\mathbf{X}_{<k>}^{\rho*}$.

\begin{figure}[t]
	\vspace{-0.1cm}
	\centering
	\includegraphics[width=0.9\linewidth,  trim={0cm 0cm 0cm 0cm}, clip]{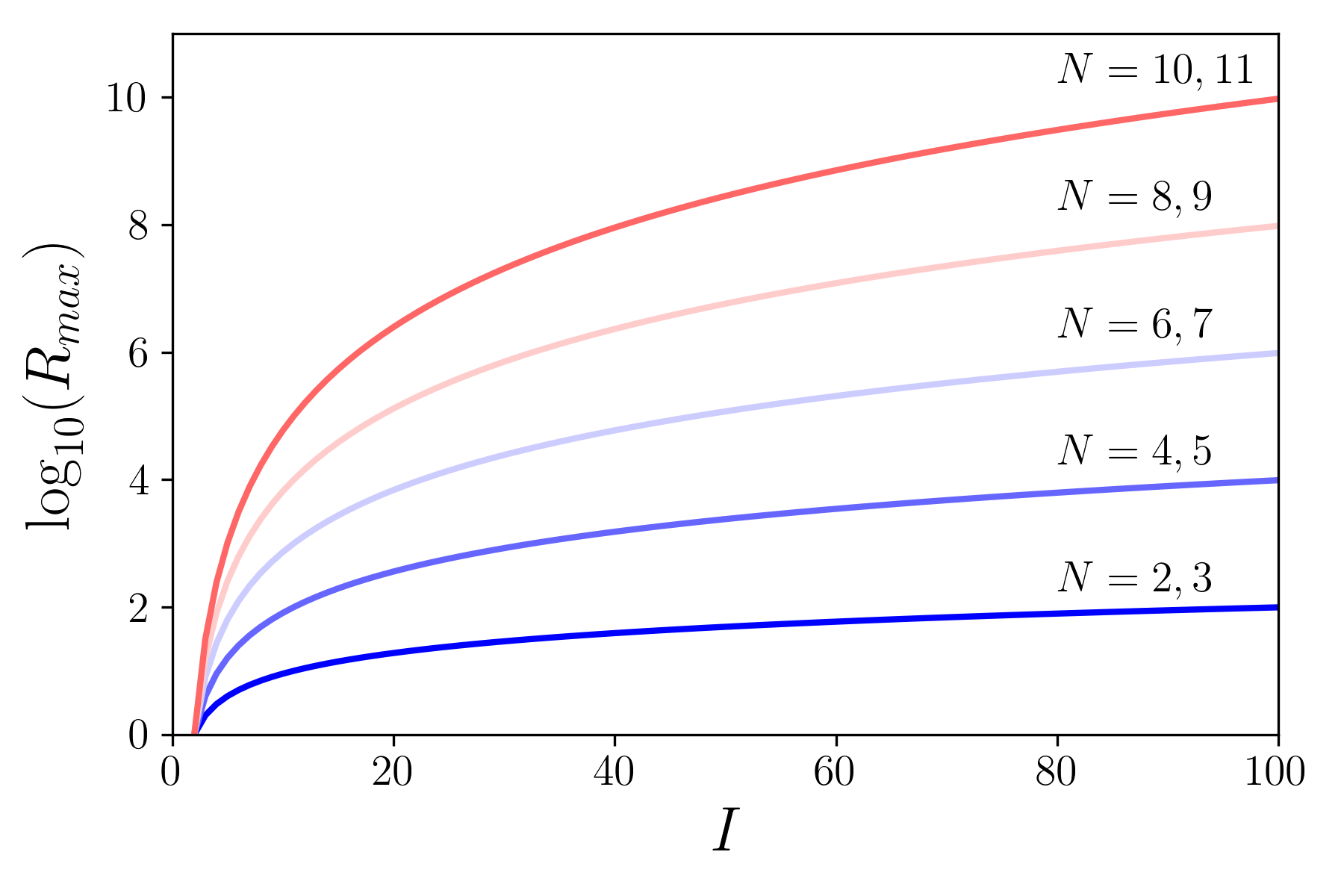}
	\vspace{-0.1cm}
	\caption{Rank detectability for $N$-th order tensors, $\ten{X} \in \mathbb{R}^{I_1 \times I_2 \times \cdots \times I_N}$, with $I_n = I$ for $n=1, \dots, N$, $N=1, \dots, 11$, and increasing values of $I$.} 
	\label{fig:detectability}
	\vspace{-0.4cm}
\end{figure}

We next show that the practical implications of Theorem \ref{th:rank} are two-fold. Firstly, given an $N$-th order tensor, $\ten{X} \in \mathbb{R}^{I_1 \times I_2 \times \cdots \times I_N}$, of rank $R$, if $R>\rank(\mathbf{X}_{<k>}^{\rho*})$, i.e. if the maximal unfolding $\mathbf{X}_{<k>}^{\rho*}$ is full rank, a search for $R$ may commence from $\rank(\mathbf{X}_{<k>}^{\rho*})$. Secondly, and perhaps more importantly, if $\mathbf{X}_{<k>}^{\rho*}$ is rank-deficient, then, generically, $R = \rank(\mathbf{X}_{<k>}^{\rho*})$, i.e. the rank is \textit{detected}, and no iterative procedure is required. The maximum detectable rank is therefore a function of the tensor dimensions, and is given by
\begin{equation}\label{eq:maxrank}
	R_{max} = \min\bigg\{ \prod_{i \in \mathcal{S}_1} I_i-1, \prod_{i \in \mathcal{S}_2} I_i-1 \bigg\}
\end{equation}
Fig. \ref{fig:detectability} illustrates the range of $R$ in (\ref{eq:maxrank}) for $N$-th order tensors, whereby equal dimensionality across modes has been assumed for simplicity, such that $R_{max} = I^{\floor{\frac{N}{2}}}-1$  and $I_i=I$ for all $i=1, \dots, N$. For such configuration, the maximum detectable rank is exponential in $\floor{\frac{N}{2}}$. Consider for example an $N$-th order tensor $\ten{X} \in \mathbb{R}^{20 \times 20 \times \cdots \times 20}$. If $N=2,3$, $R_{max} = 19$, while if $N=4,5$, $R_{max}=399$ and $R_{max}=7999$ if $N=6,7$. This means that, generically, the tensor rank, $R$, can immediately be found if $\rank(\mathbf{X}_{<k>}^{\rho*})\leq R_{max}$, i.e. if the maximal unfolding is rank-deficient, which removes the need for iterative procedures. The simplicity of such test makes it very useful in practice. 

%
%
%

\section{Conclusions}

A tight lower bound of the detectable tensor rank, $R$, has been derived which generalises to tensors of any order. This has been achieved with the purpose of easing computational costs on the Canonical Polyadic Decomposition (CPD), by considering the flattened version of the CPD and its corresponding factors. The desired bound of $R$ has been shown to be equal to the maximum possible matrix rank attainable from any unfolding of the tensor. Determining this so-called maximal unfolding reduces to solving the number partitioning problem on the indices of the original tensor, which can be done efficiently through dynamic programming. Finally, for a generic tensor, the bound has been demonstrated to be attainable when the maximal unfolding is rank-deficient, thereby providing a simple test for rank detectability. Numerical examples illustrate the practical benefits of this result.


\ifCLASSOPTIONcaptionsoff
\newpage
\fi



%
\balance
\bibliographystyle{IEEEtran}
\bibliography{references_pepe}

%




\end{document}